\documentclass[a4paper,UKenglish,cleveref, autoref, thm-restate]{lipics-v2021}

\pdfoutput=1
\hideLIPIcs
\usepackage{amsmath}
\usepackage{amsthm}
\usepackage{amssymb}
\usepackage{amsfonts}
\usepackage{graphicx}
\usepackage{makecell}
\usepackage{bm}
\usepackage{dcolumn}
\usepackage{float}
\usepackage{multirow}
\usepackage{braket}
\usepackage{color}
\usepackage{tcolorbox}
\usepackage{algorithm}
\usepackage{algorithmic}

\usepackage{tikz}
\usetikzlibrary{quantikz2}
\usepackage{appendix}
\usepackage{mathtools}
\usepackage{threeparttable}
\DeclarePairedDelimiter\rbra{\lparen}{\rparen}
\DeclarePairedDelimiter\sbra{\lbrack}{\rbrack}
\DeclarePairedDelimiter\cbra{\{}{\}}
\DeclarePairedDelimiter\abs{\lvert}{\rvert}
\DeclarePairedDelimiter\Abs{\lVert}{\rVert}
\DeclarePairedDelimiter\ceil{\lceil}{\rceil}
\DeclarePairedDelimiter\floor{\lfloor}{\rfloor}

\DeclarePairedDelimiterX{\pt}[1](){#1} 

\newcommand{\bo}{O\pt}

\newcommand{\om}{\Omega\pt}

\newcommand{\ta}{\Theta\pt}

\title{Nearly Optimal Circuit Size for Sparse Quantum State Preparation} 

\author{Lvzhou Li}{Institute of Quantum Computing and Software, School of Computer Science and Engineering, Sun Yat-sen University, Guangzhou 510006, China}{lilvzh@mail.sysu.edu.cn}{https://orcid.org/0000-0001-5941-7036}{}

\author{Jingquan Luo}{Institute of Quantum Computing and Software, School of Computer Science and Engineering, Sun Yat-sen University, Guangzhou 510006, China}{luojq25@mail2.sysu.edu.cn}{https://orcid.org/0009-0003-7757-9858}{}

\authorrunning{L. Li and J. Luo} 

\Copyright{Lvzhou Li and Jingquan Luo} 

\funding{Lvzhou Li and Jingquan Luo were supported by the National Key Research and Development Program of China (Grant No.2024YFB4504004), the National Natural Science Foundation of China (Grants No.92465202 and No.62272492), the Guangdong Provincial Quantum Science Strategic Initiative (Grant No.GDZX2303007), and the Guangzhou Science and Technology Program (Grant No.2024A04J4892).}



\nolinenumbers

\begin{document}

\maketitle

\begin{abstract}
Quantum state preparation is a fundamental and significant  subroutine in quantum computing. In this paper, we conduct a systematic investigation on the circuit size (the total count of elementary gates in the circuit) for sparse quantum state preparation. A quantum state is said to be $d$-sparse if it has only  $d$ non-zero amplitudes. 
For the task of preparing an $n$-qubit $d$-sparse quantum state, we obtain the following results:
\begin{itemize}
   \item \textbf{Without ancillary qubits:}  
   Any $n$-qubit $d$-sparse quantum state can be prepared by a quantum circuit of size $O(\frac{nd}{\log n} + n)$ without using ancillary qubits, which improves the previous best results. It  is asymptotically optimal when $d = \mathrm{poly}(n)$, and this optimality holds for a broader scope under some reasonable assumptions.
 \item \textbf{With limited ancillary qubits:} (i) Based on the first result, we prove for the first time a trade-off between the number of ancillary qubits and the circuit size: any $n$-qubit $d$-sparse quantum state can be prepared by a quantum circuit of size $O(\frac{nd}{\log (n + m)} + n)$ using $m$ ancillary qubits for any $m \in O(\frac{nd}{\log nd} + n)$. 
  (ii) We establish a matching lower bound  $\Omega(\frac{nd}{\log {(n + m)} }+ n)$ under some reasonable assumptions, and  obtain a slightly weaker lower bound $\Omega(\frac{nd}{\log {(n + m)} + \log d} + n)$ without any assumptions.

 \item   \textbf{With unlimited ancillary qubits:} Given arbitrary amount of ancillary qubits available, the circuit size for preparing $n$-qubit $d$-sparse quantum states is $\Theta(\frac{nd}{\log nd} + n)$.
\end{itemize}

\end{abstract}

\section{Introduction}

Since the inception of quantum computing \cite{feynman2018simulating}, an increasing number of quantum algorithms have emerged, exhibiting acceleration advantages over classical algorithms. Notable examples include  Shor's algorithm \cite{shor1994algorithms}, Grover's algorithm \cite{grover1996fast}, HHL algorithm \cite{harrow2009quantum}, and Hamiltonian simulation algorithms \cite{childs2018toward, low2017optimal, low2019hamiltonian, berry2015simulating}, as well as quantum machine learning algorithms \cite{lloyd2014quantum, kerenidis2016quantum, kerenidis2019q, kerenidis2021quantum, rebentrost2014quantum}. Within these algorithms, the preparation of quantum states plays a pivotal role. Usually, the first and inevitable step to process classical data by quantum algorithms is to encode the data into quantum states. If this step consumes substantial resources, it may offset the acceleration advantages of quantum algorithms. Several works have indicated that the substantial accelerations achieved by quantum machine learning stem from the underlying input assumptions \cite{tang2019quantum, chia2019quantum, tang2021quantum, chia2022sampling, gilyen2022improved}. Therefore, how to efficiently prepare quantum states is a fundamental and significant issue in the field of quantum computing.

Before delving into the discussion of quantum state preparation, it is essential to specify the elementary gate set and complexity metrics. In this paper, we adopt the most common gate set, which consists of single-qubit gates and CNOT gates, enabling the accurate implementation of any unitary transformation~\cite{barenco1995elementary}. Additionally, for convenience, the Toffoli gate is permitted, which can be constructed using 10 single-qubit gates and 6 CNOT gates \cite{nielsen2001quantum}. Various standards can be employed to gauge the efficiency of a quantum circuit, including size (the total count of elementary gates), depth (the number of layers such that gates in the same layer can be run in parallel), space (the number of ancillary qubits), and more. Given the considerably higher implementation cost of CNOT gates compared to single-qubit gates, the count of CNOT gates is also utilized as a measure of algorithmic efficiency.

When the quantum state to be prepared does not possess any structural features, we refer to it as the general quantum state preparation problem. This problem has already been extensively researched, and current preparation algorithms have achieved the optimal circuit size $\ta{2^n}$~\cite{grover2002creating,shende2005synthesis,bergholm2005quantum,plesch2011quantum,iten2016quantum}.
Given the high cost of preparing general quantum states and the fact that, in practical applications, data often exhibits special structures, there has been considerable literature focusing on the preparation of special quantum states. These include states whose amplitudes are given by a continuous function \cite{holmes2020efficient, gonzalez2024efficient, rattew2022preparing, marin2023quantum, mcardle2022quantum, moosa2023linear}, states whose amplitudes are accessed through a black-box oracle \cite{sanders2019black, bausch2022fast, wang2021fast}, and states under the low-rank assumption~\cite{araujo2023low}. Another class of quantum states that holds both practical and theoretical significance is sparse quantum states \cite{gleinig2021efficient, malvetti2021quantum, ramacciotti2023simple, mozafari2022efficient, de2020circuit, de2022double, zhang2022quantum, sun2023asymptotically, mao2024towards}. 

An $n$-qubit $d$-sparse quantum state refers to an $n$-qubit quantum state with  $d$ non-zero amplitudes. Given two positive integers $n$ and $0 < d \leq 2^n$, along with a set
\begin{align}
    \mathcal{P} =\cbra*{\rbra{\alpha_i, q_i}}_{0 \leq i \leq d-1}
\end{align}
such that $\alpha_i \in \mathbb{C}$, $q_i \in \cbra{0, 1}^n$ for all $0 \leq i \leq d-1$, $\sum_i |\alpha_i|^2 = 1$ and $q_i \neq q_j$ for all $i \neq j$, the aim of  sparse quantum state preparation (SQSP) is to generate a quantum circuit which acts on $n + m$ qubits and implements a unitary operator $U$ satisfying
\begin{align}
    U \ket{0}^{\otimes n}\ket{0}^{\otimes m} = \sum_{i = 0}^{d-1} \alpha_i \ket{q_i} \ket{0}^{\otimes m},
\end{align}
where $m \geq 0$ is an integer and the last $m$ qubits serve as ancillary qubits.

The importance of SQSP is self-evident, and the motivation is twofold:
\begin{itemize}
\item  First, many practically relevant states in quantum computing and processing have the property that only a small proportion of the basis states have nonzero coefficients.
 Some prominent examples of sparse states are W states, GHZ states,  generalized Bell states, and thermofield double states \cite{cottrell2019build}. 

\item Second,  SQSP offers a finer-grained perspective on quantum state preparation. For the general quantum state preparation problem, the circuit size is proven to be $\ta{2^n}$, which does not consider the dependence on the parameter $d$. In other words, it simply assumes $d=2^n$. A more fine-grained complexity should consider this dependence, and recover the general case when $d=2^n$.
\end{itemize}

\subsection{Contributions} 
In this paper, we conduct a systematic investigation of the circuit size of sparse quantum state preparation, exploring the three scenarios:  without, with limited, and with unlimited ancillary qubits. The complexity notations $O$, $\Omega$, $o$ and $\omega$ will be explained in detail in Section \ref{section:pre}. 

\subsubsection{SQSP without Ancillary Qubits}First, we consider the scenario of not using ancillary qubits.

\begin{theorem}\label{thm:withoutancillary}
   Any $n$-qubit $d$-sparse quantum state can be prepared by a quantum circuit of size $O(\frac{nd}{\log n} + n)$ without using ancillary qubits.
\end{theorem}

  Previously, the best result without using ancillary qubits  \cite{gleinig2021efficient, malvetti2021quantum} can only achieve the circuit size $O(nd)$. We demonstrate for the first time that any $n$-qubit $d$-sparse quantum state can be prepared by a circuit of size $o(nd)$ without using ancillary qubits.
Moreover, in \cref{thm: umlimited} we will see that the upper bound $O(\frac{nd}{\log n} + n)$ is asymptotically optimal when $d = \mathrm{poly}(n)$,  and in \cref{thm:lb2} we get the lower bound $\Omega(\frac{nd}{\log n} + n)$ when $m = 0$, which indicates that the upper bound is asymptotically optimal when $d \leq 2^{\delta n}$ for a sufficiently small constant $\delta \in \rbra{0, 1}$ under reasonable assumptions.

\subsubsection{\texorpdfstring{SQSP with $m$ Ancillary Qubits}{SQSP with m Ancillary Qubits}} 
Next, we consider the scenario of using $m$ ancillary qubits. To the best of our knowledge, we are the first to consider and demonstrate the trade-off between the number of ancillary qubits and the circuit size of SQSP. Based on \cref{thm:withoutancillary}, we obtain the following theorem:

\begin{theorem}\label{thm:tradeoff}
   For any $m \in \bo{\frac{nd}{\log nd}+n}$, any $n$-qubit $d$-sparse quantum state can be prepared by a quantum circuit of size $\bo{\frac{nd}{\log (m+n)} + n}$ using $m$ ancillary qubits.
\end{theorem}

\begin{remark}\label{remark:ub-on-ancillary}
    In \cref{thm:tradeoff},   when $m = \ta{\frac{nd}{\log nd}+n}$, we get  $\bo{\frac{nd}{\log nd}+n}$ on the circuit size. Actually, more ancillary qubits than $\ta{\frac{nd}{\log nd}+n}$ are useless for reducing the upper bound $\bo{\frac{nd}{\log nd}+n}$. Otherwise, it would contradict the fact that the number of ancillary qubits must be asymptotically less than or equal to the circuit size. Since only one-qubit gates and CNOT gates are considered as elementary quantum gates here, if the number of ancillary qubits is more than twice as much as the circuit size, there must exist unused ancillary qubits. Therefore, without loss of generality, we can assume that $m \in \bo{\frac{nd}{\log nd}+n}$.
\end{remark}

We demonstrate that the upper bound established in \cref{thm:tradeoff} is tight under reasonable assumptions. 
Notice that all sparse quantum state preparation methods proposed so far use at most $\bo{d}$ arbitrary-angle single-qubit rotation gates (i.e., $R_x(\theta)$, $R_y(\theta)$, $R_z(\theta)$). For example, in permutation-based  algorithms~\cite{malvetti2021quantum, ramacciotti2023simple}, arbitrary-angle single-qubit rotation gates are only used in the first step to prepare a $ \lceil\log d\rceil $-qubit quantum state, while in the second step, implementing a permutation only involves CNOT gates and $\bo{n}$ specific types of single-qubit gates. We believe that using $\omega(d)$ arbitrary-angle single-qubit rotation gates is pointless. 

\begin{theorem}\label{thm:lb2}
    Suppose $d \leq 2^{\epsilon n}$ for a sufficiently small constant $\epsilon \in \rbra{0, 1}$, and given $m$ ancillary qubits available, if an algorithm $\mathcal{A}$ for preparing $n$-qubit $d$-sparse quantum states satisfies the following conditions:
    \begin{enumerate}
        \item $\mathcal{A}$ uses at most $\bo{d}$ single-qubit rotation gates with arbitrary angles,
        \item all other single-qubit gates amount to a total of $O(n)$ types,
    \end{enumerate}
 then $\mathcal{A}$ requires $\Omega(\frac{nd}{\log(m+n)} + n)$ elementary quantum gates in the worst case.
\end{theorem}

More specifically, excluding rotation gates with arbitrary angles, the single-qubit gates considered in \cref{thm:lb2} include NOT gates and Phase$(\pm \pi/2^j)$ gates for $1 \leq j \leq n$, where the latter are useful for implementing an $n$-Toffoli gate without ancillae \cite{gidney2015}.

\begin{remark}
    When only $m = \mathrm{poly}(n)$ ancillary qubits are available, the lower bound in \cref{thm:lb2} is $\om{\frac{nd}{\log n}+n}$, which can be achieved without ancillary qubits as shown in \cref{thm:withoutancillary}. Therefore, $\mathrm{poly}(n)$ ancillary qubits do not help in reducing the circuit size of SQSP under some reasonable assumptions.
\end{remark}

Without any assumptions, we prove a slightly weaker lower bound on the circuit size as follows.
\begin{theorem}\label{thm:lb1}
   Given $m$ ancillary qubits available, there exist $n$-qubit $d$-sparse quantum states such that any algorithm to prepare them requires $\om{\frac{nd}{\log(m + n) + \log d} + n}$ elementary quantum gates.
\end{theorem}

\subsubsection{SQSP with Unlimited Ancillary Qubits}
We provide a complete characterization of the circuit size for SQSP in the case where an unlimited number of ancillary qubits is available.

\begin{theorem}\label{thm: umlimited}
    {With unlimited  number of ancillary qubits available}, the circuit size for preparing $n$-qubit $d$-sparse quantum states is $\ta{\frac{nd}{\log nd} + n}$. Furthermore, if $d =\mathrm{poly}(n)$, then the circuit size is $\ta{\frac{nd}{\log n} + n}$, which can be achieved without using ancillary qubits.
\end{theorem}

\begin{proof}

According to \cref{remark:ub-on-ancillary}, it is sufficient to assume  $m = \ta{\frac{nd}{\log nd}+n}$ ancillary qubits available. Therefore, by setting $m = \ta{\frac{nd}{\log nd}+n}$ in \cref {thm:lb1}, we get the lower bound   $\om{\frac{nd}{\log nd} + n}$, which can be achieved by \cref{thm:tradeoff}.

Furthermore,  if $d = \mathrm{poly} (n)$, then  $\ta{\frac{nd}{\log nd} + n}$ becomes $\ta{\frac{nd}{\log n} + n}$, which can be achieved without using ancillary qubits as shown in  \cref{thm:withoutancillary}.
\end{proof}

\subsection{Proof Techniques}

\subsubsection{SQSP  without Ancillary Qubits}

The idea of  SQSP  without ancillary qubits proposed in this paper aligns with prior research~\cite{malvetti2021quantum, ramacciotti2023simple}. Given the state description  $\mathcal{P} = \cbra{(\alpha_i, q_i)}_{0\leq i \leq d-1}$, initially a dense $\ceil{\log d}$-qubit quantum state $\sum_i \alpha_i \ket{\sigma^{-1}(q_i)}$ is prepared, where $\sigma$ is some permutation such that $0 \leq \sigma^{-1}(q_i) \leq d - 1$. Subsequently, this state is transformed into the target quantum state $\sum_i \alpha_i \ket{q_i}$ by applying $\sigma$. The dense quantum state  $\sum_i \alpha_i \ket{\sigma^{-1}(q_i)}$ can be prepared with a circuit of size $O(d)$. Therefore, the primary contribution to circuit size arises from the second step, i.e., the implementation of the permutation $\sigma$. Prior works \cite{malvetti2021quantum, ramacciotti2023simple} have adopted circuits of size $O(nd)$ for implementing $\sigma$, which results in an overall circuit size of $O(nd)$. In comparison, we develop a method for implementing $\sigma$ with a size of $\bo{\frac{nd}{\log n}+n}$, as detailed in \cref{lemma:sigma}.

\begin{lemma}\label{lemma:sigma}
For any permutation $\sigma \in \mathfrak{S}_{2^n}$, there exists a quantum circuit implementing $\sigma$ of size $\bo{\frac{n\,\mathrm{size}(\sigma)}{\log n} + n\log\min\cbra{\mathrm{size}(\sigma), \log n}}$ without using ancillary qubits. 
\end{lemma}

In the above, $\mathfrak{S}_{2^n}$ denotes the set of all permutations on $\{0, 1\}^n$, and $\mathrm{size}(\sigma)$ denotes the number of non-fixed points of a permutation $\sigma$ (see the formal definition in \cref{eq:non-fixed}).  
Reversible circuits, which implement permutations, have been extensively studied in the literature~\cite{shende2003synthesis, brodsky2004reversible, wille2009bdd, saeedi2010reversible, saeedi2010library, saeedi2013synthesis, abdessaied2014upper, zakablukov2017asymptotic, li2023asymptotically, markov2008optimal, jiang2020optimal}.  
Shende et al.~\cite{shende2003synthesis} proposed a synthesis method for realizing a permutation $\sigma \in \mathfrak{S}_n$ using a circuit of size $\bo{n2^n}$; a similar result was obtained in~\cite{brodsky2004reversible}.  
Saeedi et al.~\cite{saeedi2010reversible, saeedi2010library} introduced cycle-based methods that reduced the circuit size in practice, although the asymptotic complexity remained unchanged.  
Significant improvements were later made independently by Zakablukov~\cite{zakablukov2017asymptotic} and  Wu and Li~\cite{li2023asymptotically}, who achieved an asymptotically optimal circuit size of $\bo{\frac{n2^n}{\log n}}$.  
However,  Refs. ~\cite{zakablukov2017asymptotic,li2023asymptotically} did not consider the sparsity of permutations\footnote{The sparsity of a permutation $\sigma \in \mathfrak{S}_{2^n}$ is measured by the quantity $\mathrm{size}(\sigma)$, which equals $2^n$ in the worst case, but can be significantly smaller in many practical applications.}.  
\cref{lemma:sigma} improves upon the results in~\cite{zakablukov2017asymptotic, li2023asymptotically}, and plays a key role in reducing the circuit size from $\bo{nd}$ to $\bo{\frac{nd}{\log n} + n}$ in \cref{thm:withoutancillary}. We believe this result is of independent interest and may find applications in other contexts.

\subsubsection{\texorpdfstring{SQSP  with $m$ Ancillary Qubits}{SQSP  with m Ancillary Qubits}}
To further reduce the circuit size with ancillary qubits, we introduce a new framework for SQSP that is totally different from the approach without ancillary qubits. Specifically, we employ a novel integer encoding method, denoted as the $(n, r)$-unary encoding (for the formal definition, see \cref{def:nr}). The new framework is roughly divided into two steps: firstly, prepare an intermediate state using the $(n,r)$-unary encoding, and secondly transform the intermediate state to a state in the binary encoding. Intuitively, the $(n, r)$-unary encoding divides an $n$-bit integer to $\frac{n}{r}$ continuous parts, each of length $r$, and represents each part using the corresponding unary encoding. When $r = n$, the $(n, r)$-unary encoding recovers to the standard unary encoding. This integer representation method was previously introduced in \cite{johri2021nearest}, where the authors showed that a general quantum state using the $(n, r)$-unary encoding can be prepared by a quantum circuit of size $\bo{2^n}$ and depth $\bo{r2^{n-r}}$ with $\frac{n2^r}{r}$ qubits. However, they focused on reducing the circuit depth and did not account for the sparsity of the quantum state. Also, their technique does not apply to SQSP. We demonstrate that a sparse quantum state using the $(n, r)$-unary encoding can be efficiently prepared. 
\begin{lemma}\label{lemma:nr-encoding}
    Given two positive integers $n$ and $r$ such that $r$ divides $n$, and a set of size $d$ $\cbra{(q_i, \alpha_i)}_{0 \leq i \leq d-1}$ such that $\sum_i \abs{\alpha_i}^2 = 1$, $q_i \in \cbra{0, 1}^n$ for all $0 \leq i \leq d-1$ and $q_i \neq q_j$ for any $i \neq j$, there exists a quantum circuit preparing the following $\frac{n2^r}{r}$-qubit state
    \begin{align}
        \sum_{i=0}^{d-1} \alpha_i \ket{e_{q_i(n-r, n)}} \ket{e_{q_i(n-2r, n-r)}} \dots \ket{e_{q_i(0, r)}}
        = \sum_{i=0}^{d-1} \alpha_i \rbra*{\bigotimes_{j=0}^{\frac{n}{r}-1} \ket{e_{q_i(jr, (j+1)r)}}}.
    \end{align}
    The circuit is of size $\bo{\frac{nd}{r}}$, using one ancillary qubit.
\end{lemma}
This state using the $(n, r)$-unary encoding serves as an intermediate form, ultimately transformed into the binary encoding. We believe that this intermediate state will also have applications in other tasks.

\subsubsection{Lower Bound for Circuit Size}

Prior proofs of lower bounds for quantum circuit size relied on dimensionality arguments~\cite{shende2004smaller, plesch2011quantum}, resulting in a trivial lower bound of $\om{d}$  when applied to the sparse quantum state preparation problem. Here we employ the counting argument. The challenge in applying this method to establish lower bounds for quantum circuit size lies in the existence of an infinite number of single-qubit quantum gates, i.e., the parameters of $\cbra{R_x(\theta), R_y(\theta), R_z(\theta)}$ are continuous. The fundamental idea in the proof is to discretize the parameters of the single-qubit quantum gates. Specifically, we consider approximately preparing  sparse states in the following set to sufficient precision, ensuring that the resulting states are distinct from one another:
\begin{align}
    \mathcal{D}_d \coloneqq \left\{ \ket{\psi_\mathcal{S}} \coloneqq \frac{1}{\sqrt{d}} \sum_{i \in \mathcal{S}} \ket{i} \mid \mathcal{S} \subset \{0, \dots, 2^n-1\}, \abs{\mathcal{S}} = d \right\}.
\end{align}
To achieve this, we only need to discretize the single-qubit gates to a certain precision, and we demonstrate that this discretization does not increase the circuit size. Consequently, we can derive a lower bound on the circuit size for SQSP based on the lower bound for the task of approximately preparing the states in $\mathcal{D}_d$ using the discretized single-qubit gates and CNOT gates. We hope the discretization technique will find more applications in proving lower bounds on the circuit size for other subclasses of quantum states.

\subsection{Related Work}

The problem of SQSP  has garnered significant attention \cite{gleinig2021efficient, malvetti2021quantum, ramacciotti2023simple, mozafari2022efficient, de2020circuit, de2022double, zhang2022quantum, sun2023asymptotically, mao2024towards}. 
  Gleinig and Hoefler \cite{gleinig2021efficient} were the first to address this problem and presented an algorithm generating a quantum circuit of size $O(nd)$ without using ancillary qubits for any $n$-qubit $d$-sparse quantum state. 
Subsequently, Malvetti et al.~\cite{malvetti2021quantum} and Ramacciotti et al.~\cite{ramacciotti2023simple} presented permutation-based preparation algorithms, achieving the same circuit size using 0 and 1 ancillary qubits, respectively.
 de Veras et al.~\cite{de2020circuit} improved upon the ideas in \cite{park2019circuit, trugenberger2001probabilistic} and introduced a preparation algorithm based on quantum state splitting, achieving the same circuit size by employing $2$ ancillary qubits.  They further optimized this algorithm \cite{de2022double}, presenting a circuit size depending on the Hamming weights of the basis states with non-zero amplitudes: $O(\sum_i |{q}_i|)$. In the worst case, the circuit size remains $O(nd)$. Mozafari et al.~\cite{mozafari2022efficient} employed decision diagrams to represent quantum states and introduced an algorithm with circuit size $O(kn)$, using $1$ ancillary qubit, where $k$ denotes the number of paths in the decision diagram. For a $d$-sparse quantum state, as the number of paths in its associated decision diagram satisfies $k \leq d$, the circuit size of this method in the worst-case remains $O(nd)$. In summary, the circuit sizes of all the aforementioned algorithms are $O(nd)$.  Recently, Mao et al.~\cite{mao2024towards} propose a sparse state preparation algorithm which generates  circuits of size $O(\frac{nd}{\log n} + n)$ with two ancillary qubits, improving the previous bound  $O(nd)$. 

In addition to optimizing circuit size, Zhang et al.~\cite{zhang2022quantum} focused on reducing circuit depth. They showed that with $O(nd\log d)$ ancillary qubits and $O(nd\log d)$ gates\footnote{The circuit size was not explicitly stated in~\cite{zhang2022quantum}, but can be inferred from the algorithm presented therein.}, any $n$-qubit $d$-sparse quantum state can be prepared by a circuit of depth $O(\log nd)$. They also proved a matching lower bound of $\om{\log nd}$ on the depth.  
Sun et al.~\cite{sun2023asymptotically} focused on circuit depth for general quantum state preparation, and noted that their method can also be extended to sparse states, yielding a depth complexity of $O(n\log nd + \frac{nd^2\log d}{n + m})$ with $m$ ancillary qubits. This was later improved to $\bo{\frac{nd\log d\log m}{m} + \log nd}$ in~\cite{zhang2024circuit}.

\subsection{Discussion}

We achieve a relatively complete understanding of the circuit size for SQSP in this paper; however, several interesting problems remain open and are worthy of further discussion:

\begin{itemize}
\item   In the scenario of using $m$ ancillary qubits,  can we prove a matching lower bound $\bo{\frac{nd}{\log (m+n)} + n}$ on the circuit size without assumptions?

\item  In this paper, all of our results focus on exact preparation, meaning the precise preparation of the states. What about the approximate preparation of $d$-sparse quantum states?

\item What is the optimal trade-off between the number of ancillary qubits and the circuit depth for SQSP? To the best of our knowledge, the only known results are a rough upper bound of $O(n\log nd + \frac{nd^2\log d}{n + m})$ mentioned in~\cite{sun2023asymptotically}, and a refined bound of $\bo{\frac{nd\log d\log m}{m} + \log nd}$ given in~\cite{zhang2024circuit}.

\end{itemize}

\subsection{Organization}
In Section~\ref{section:pre}, we recall some notations.
Section~\ref{sec:without} is devoted to the proof of Theorem~\ref{thm:withoutancillary}.
In Section~\ref{sec:with}, we prove Theorems~\ref{thm:tradeoff}, \ref{thm:lb2} and \ref{thm:lb1}.
Some conclusions are made in Section \ref{sec:conclusion}.

\section{Preliminaries}\label{section:pre}

In this paper, all logarithms are base 2. A permutation $\sigma$ on a finite set $A$ is a bijective function $\sigma \colon A \xrightarrow{} A$. Given a positive integer $n$, we denote by $\mathfrak{S}_{2^n}$ the set of all permutations on  $\{0, 1\}^n$, and the permutations considered in this paper all belong to $\mathfrak{S}_{2^n}$. 
A cycle $\rbra{a_1, a_2, \dots, a_k}$ is a permutation $\sigma$ such that $\sigma\rbra{a_1} = a_2$, $\sigma\rbra{a_2} = a_3$, $\dots$, and $\sigma\rbra{a_k} = a_1$. The length of a cycle is the number of elements it contains. A cycle of length two is called a transposition. A cycle of length $k$ is called a $k$-cycle. The composition of two permutations $\sigma_1$ and $\sigma_2$ is denoted by $\sigma_2\circ \sigma_1$ where the right one is applied first. The composition operation is typically not commutative, but when two permutations are disjoint, we can interchange them. 

 Any permutation can be decomposed as a composition of a finite number of transpositions. A permutation is even (odd) if it can be written as a composition of an even (odd) number of transpositions. 
 
 \begin{lemma}[\cite{hall2018theory}]\label{lm: permutation-dec}
 Any permutation can be written as the composition of a finite number of pairwise disjoint cycles.
 \end{lemma}
 
 A fixed point of a permutation $\sigma$ is an element $x \in \{0, 1\}^n$ satisfying $\sigma(x) = x$. 
 For any permutation $\sigma$, let $\mathrm{size}(\sigma)$ denote the number of non-fixed points of $\sigma$:
\begin{align}\label{eq:non-fixed}
    \mathrm{size}(\sigma) \coloneqq \abs{\{x \in \{0, 1\}^n \mid \sigma(x) \neq x\}}.
\end{align}

Here we give a simple example to show the notions given above. Consider the permutation $\sigma \in \mathfrak{S}_8$ given as follows:
\begin{align*}
    & \sigma(0) = 1, \quad \sigma(1) = 5, \sigma(2) = 4, \sigma(3) = 3, \\
    & \sigma(4) = 2, \quad \sigma(5) = 7, \sigma(6) = 6, \sigma(7) = 0. 
\end{align*}
$\sigma$ could be written as 
\begin{align}
    \sigma & = (0, 1, 5, 7) \circ (2, 4) \circ (3) \circ (6) \label{equ:1} \\
    & =  (0, 1, 5, 7) \circ (2, 4) \label{equ:2} \\
    & = (0, 5) \circ (0, 1) \circ (5, 7) \circ (2, 4). \label{equ:3}
\end{align}
In \cref{equ:1}, $\sigma$ is written as the composition of pairwise disjoint cycles, where $(0, 1, 5, 7)$ is a $4$-circle and $(2, 4)$ is a transposition. In \cref{equ:2}, we omit $1$-circles. In \cref{equ:3}, $\sigma$ is decomposed as the composition of $5$ transpositions, thus $\sigma$ is an odd permutation. The number of non-fixed points of $\sigma$ is $\mathrm{size}(\sigma)=6$.

An $n$-qubit Toffoli gate is a quantum gate acting on $n$ qubits and applying an X gate on the last qubit when the preceding $n-1$ qubits are all in the state $\ket{1}$. Gidney \cite{gidney2015} demonstrated the following lemma in his well-known blog:

\begin{lemma}[\cite{gidney2015}]\label{lemma:gidney}
    Any $n$-qubit Toffoli gate can be implemented by a quantum circuit of size $\bo{n}$ without ancillary qubits.
\end{lemma}

We briefly introduce several complexity notations. Given two functions $f(n) $ and $g(n) $ defined on $\mathbb{N}$:
\begin{itemize}
   \item $f(n) = \bo{g(n)} $ if there exist positive constants $c $ and $n_0$ such that $f(n) \leq c \cdot g(n) $ for all $n \geq n_0 $.  
   \item $f(n) = \om{g(n)} $ if there exist positive constants $c $ and $n_0$ such that $f(n) \geq c \cdot g(n) $ for all $n \geq n_0$.  
   \item $f(n) = \ta{g(n)} $ if both $f(n) = \bo{g(n)} $ and $f(n) = \om{g(n)} $.
     \item $f(n) = o(g(n))$ if, for any positive constant $ c > 0 $, there exists a constant $ n_0 $ such that $f(n) \leq c \cdot g(n)$ for all $n \geq n_0$.
    \item $f(n) = \omega(g(n))$ if, for any positive constant $ c > 0 $, there exists a constant $ n_0 $ such that $f(n) \geq c \cdot g(n)$ for all $n \geq n_0$.
\end{itemize}

\section{SQSP without Ancillary Qubits}\label{sec:without}

In this section, we present a sparse quantum state preparation algorithm without using ancillary qubits. The algorithm presented here relies on the efficient  implementation of a permutation $\sigma \in \mathfrak{S}_n$. 
As mentioned earlier, the circuit size for implementing $\sigma$ in prior works~\cite{malvetti2021quantum, ramacciotti2023simple} is $\bo{n\,\mathrm{size}(\sigma)}$. Improving upon the results of~\cite{zakablukov2017asymptotic, li2023asymptotically}, we show in the following lemma that this bound can be further reduced.

\begin{lemma}[Restatement of \cref{lemma:sigma}]\label{lemma:sigma-restate}
For any permutation $\sigma \in \mathfrak{S}_{2^n}$, there exists a quantum circuit implementing $\sigma$ of size $\bo{\frac{n\,\mathrm{size}(\sigma)}{\log n} + n\log\min\cbra{\mathrm{size}(\sigma), \log n}}$ without using ancillary qubits. 

\end{lemma}

\begin{proof}

The basic steps  to implement $\sigma$ are as follows:
\begin{itemize}
    \item \textbf{Step 1}: Decompose $\sigma$ into the composition of as few permutations $\sigma_i$ as possible, where each permutation $\sigma_i$ consists of at most $m$ pairwise disjoint transpositions and satisfies that $\mathrm{size}(\sigma_i)$ is a power of 2. $m$ is a parameter to be determined, and will be a power of 2.\label{step1}
    
    \item \textbf{Step 2}: Assume $\sigma_i = (x_0, x_1) \circ  (x_2, x_3) \circ \cdots \circ (x_{2m-2}, x_{2m-1})$ and implement each $\sigma_i$ successively by the following steps: 
    \begin{itemize}
        \item \textbf{Step 2a}: Execute the permutation $\sigma_{i,1}$ satisfying the condition: for any $0 \leq j \leq 2m-1$, $\sigma_{i,1}(x_j) = j$;
        \item \textbf{Step 2b}: Execute the permutation $\sigma_{i,2} = (0, 1) \circ (2, 3) \circ \cdots \circ (2m-2, 2m-1)$;
        \item \textbf{Step 2c}: Execute the inverse of permutation $\sigma_{i,1}$.
    \end{itemize}
\end{itemize}
Below, we illustrate how to implement this process, set the parameter $m$, and analyze the circuit size.

In Step 1, we aim to decompose $\sigma$ into a composition of at most
\begin{align}
    M = \left\lfloor \frac{\mathrm{size}(\sigma)}{m} \right\rfloor + \bo{\log \left(\min\{m, \mathrm{size}(\sigma)\}\right)}
\end{align}
permutations $\{\sigma_i \mid 0 \leq i \leq M-1\}$, i.e., $\sigma = \sigma_{M-1} \circ \dots \circ \sigma_0$, where each $\sigma_i$ consists of at most $m$ pairwise disjoint transpositions, and $\mathrm{size}(\sigma_i)$ is a power of $2$.
Firstly, by \cref{lm: permutation-dec}, we decompose $\sigma$ into the composition of $K$ disjoint cycles for some integer $K \geq 0$, i.e., $\sigma = \rho_{K-1} \circ \dots \circ \rho_0$, where each $\rho_k$ is an $r_k$-cycle for some integer $r_k > 0$. It is clear that $\sum_k r_k = \mathrm{size}(\sigma)$.
Each $r_k$-cycle can be decomposed into two sets of disjoint transpositions, each of size either $\left\lfloor \frac{r_k}{2} \right\rfloor$ or $\left\lfloor \frac{r_k}{2} \right\rfloor - 1$~\cite{moore2001parallel}. For example, suppose $\rho = (x_0, x_1, \dots, x_{2k-1})$ is a cycle of even length. Then we have $\rho = \rho'' \circ \rho'$, where
\begin{align}
    \rho' &= (x_0, x_{2k-1}) \circ (x_1, x_{2k-2}) \circ \cdots \circ (x_{k-1}, x_{k}), \\
    \rho'' &= (x_1, x_{2k-1}) \circ (x_2, x_{2k-2}) \circ \cdots \circ (x_{k-1}, x_{k+1}).
\end{align}
A similar construction applies when $\rho$ is an odd-length cycle.
Thus, $\sigma$ can be decomposed into two sets of disjoint transpositions, each of size at most $\floor{\frac{\mathrm{size}(\sigma)}{2}}$.
Each such set can be further partitioned into at most $\floor{\frac{\mathrm{size}(\sigma)}{2m}} + \ceil{\log (\min\{m, \mathrm{size}(\sigma)\})}$ sets of disjoint transpositions, where each set contains at most $m$ transpositions and the number of transpositions is a power of $2$.

In Step 2, we show each $\sigma_i$ can be efficiently implemented. To avoid introducing more parameters, in the following we assume that $\sigma_i = (x_0, x_1) \circ  (x_2, x_3) \circ \cdots \circ (x_{2m-2}, x_{2m-1})$ is composed of $m$ pairwise disjoint transpositions. Step 2 is divided into three substeps. First, we illustrate how to implement a permutation $\sigma_{i,1}$ that satisfies $\sigma_{i,1}(x_j) = j$ for any $0 \leq j \leq 2m-1$. Note that $x_j$ is an integer represented by $n$ bits. We treat $x_j$ as an $n$-dimensional row vector $x_j = (x_{j, 0}, x_{j, 1}, \dots, x_{j, n-1})$ satisfying $x_j= \sum_{k=0}^{n-1} x_{j, k}2^k$, where $x_{j, k}$ represents the $k$-th bit and the bits are arranged from the least significant bit to the most significant bit, contrary to the usual convention. We construct a matrix composed of $x_j$ to track the changes in $x_j$:
\begin{align}
    A 
     = \rbra*{\begin{matrix}
        x_0 \\ x_1 \\ \dots \\ x_{2m-2} \\ x_{2m-1}
    \end{matrix} }
     = \rbra*{\begin{matrix}
        x_{0, 0} & x_{0, 1} & \dots & x_{0, n-2} & x_{0, n-1} \\
        x_{1, 0} & x_{1, 1} & \dots & x_{1, n-2} & x_{1, n-1} \\
        \vdots &&\vdots&&\vdots \\
        x_{2m-2, 0} & x_{2m-2, 1} & \dots & x_{2m-2, n-2} & x_{2m-2, n-1} \\
        x_{2m-1, 0} & x_{2m-1, 1} & \dots & x_{2m-1, n-2} & x_{2m-1, n-1} \\
    \end{matrix}}
\end{align}

Now we introduce additional conditions to $m$: suppose $m$ is a power of 2 and satisfies $2m \leq \log n$. With these conditions, realizing the permutation $\sigma_{i,1}$ is equivalent to transforming the matrix $A$ into another matrix $\widetilde{A}$ using elementary gates:
\begin{align}
    \widetilde{A} = \left( \overbrace{\begin{matrix}
        0 & 0 & \dots & 0    \\
        1 & 0 & \dots & 0    \\
        \vdots & & & \vdots  \\
        0 & 1 & \dots & 1    \\
        1 & 1 & \dots & 1   \\
    \end{matrix}}^{\log 2m} \right. \quad
    \left. \overbrace{\begin{matrix}
           0 & \dots & 0 \\
           0 & \dots & 0 \\
           & \vdots& \\
           0 & \dots & 0 \\
           0 & \dots & 0 
    \end{matrix}}^{n - \log 2m} \right)
\end{align}

Suppose that the matrix $A$ has $\ell$ distinct non-zero columns. Since the matrix $A$ has $2m$ distinct rows, we have $\log 2m \leq \ell \leq 2^{2m}-1$. When the $j$-th column is identical to the $k$-th column, i.e., $x_{i, j} = x_{i, k}$ for all $0 \leq i \leq 2m-1$, and they are not all-zero columns, we apply, without loss of generality, a CNOT gate with the $j$-th qubit as the control qubit and the $k$-th qubit as the target qubit, transforming the $k$-th column into an all-zero column. This process is repeated until the matrix no longer contains identical non-zero columns.  This step requires at most $n-\ell$ CNOT gates.
Next, by using at most $\ell$ swap gates (each of which can be implemented with $3$ CNOT gates), we swap the $\ell$ non-zero columns to the first $\ell$columns, obtaining matrix $A_1$:
\begin{align}
    A_1 = \left( \overbrace{\begin{matrix}
        a_{0, 0} & a_{0, 1} & \dots & a_{0, \ell}    \\
        a_{1, 0} & a_{1, 1} & \dots & a_{1, \ell}   \\
        \vdots & & & \vdots  \\
        a_{2m-2, 0} & a_{2m-2, 1} & \dots & a_{2m-2, \ell}   \\
        a_{2m-1, 0} & a_{2m-1, 1} & \dots & a_{2m-1, \ell}   \\
    \end{matrix}}^{\ell} \right. \quad
    \left. \overbrace{\begin{matrix}
           0 & \dots & 0 \\
           0 & \dots & 0 \\
           & \vdots& \\
           0 & \dots & 0 \\
           0 & \dots & 0 
    \end{matrix}}^{n - \ell} \right)
\end{align}

Next, we proceed to transform each row of matrix $A_1$ step by step, gradually converting it to matrix $\widetilde{A}$. We start with the first row (indexed as row $0$). If $a_{0, k} = 1$ for any $0 \leq k \leq n - 1$, we apply an X gate to the $k$-th qubit. This step requires at most $\ell$ X gates. Upon completing this step, matrix $A_1$ becomes matrix $A_2$:
\begin{align}
    A_2 = \left( \overbrace{\begin{matrix}
        0 & 0 & \dots & 0    \\
        b_{1, 0} & b_{1, 1} & \dots & b_{1, \ell}   \\
        \vdots & & & \vdots  \\
        b_{2m-2, 0} & b_{2m-2, 1} & \dots & b_{2m-2, \ell}   \\
        b_{2m-1, 0} & b_{2m-1, 1} & \dots & b_{2m-1, \ell}   \\
    \end{matrix}}^{\ell} \right. \quad
    \left. \overbrace{\begin{matrix}
           0 & \dots & 0 \\
           0 & \dots & 0 \\
           & \vdots& \\
           0 & \dots & 0 \\
           0 & \dots & 0 
    \end{matrix}}^{n - \ell} \right)
\end{align}

Suppose the rows indexed as $0$ to $j-1$ have been successfully transformed, we turn our attention to row indexed as $j$, which can be divided into two scenarios:

\begin{itemize}
    \item In the case where there exists an element $b_{j, k} \neq 0$ in row $j$ with $k > \log 2m$, for each $0 \leq k' < \ell$ satisfying $k' \neq k$ and $b_{j, k'} \neq \widetilde{A}_{j, k}$, we apply a CNOT gate with the $k$-th qubit as the control qubit and the $k'$-th qubit as the target qubit. This step requires at most $\ell$ CNOT gates. 
    Subsequently, to eliminate the value $b_{j, k}$, we employ a multi-control Toffoli gate. The control qubits for this gate correspond to the non-zero elements in row $j$ of matrix $\widetilde{A}$, while the target qubit is the $k$-th qubit. This Toffoli gate can be implemented by a circuit of size $\bo{\log 2m}$ according to \cref{lemma:gidney}.

    \item In the case where there is no element $b_{j, k} \neq 0$ in row $j$ satisfying $k > \log 2m$, we first apply a multi-control Toffoli gate, where the control qubits are those corresponding to the non-zero elements in the current row, while the target qubit is the $(\log 2m + 1)$-th qubit. It is asserted that this Toffoli gate will not alter the elements in the rows indexed by  $0$ to $j-1$ of the current matrix, as otherwise, there would exist some $0 \leq j' \leq j -1$ such that the $j'$-th row is identical to the $j$-th row, contradicting the assumption that each row is distinct. This Toffoli gate contains at most $\log 2m$ control qubits. Hence, it can be implemented by a circuit of size $\bo{\log 2m}$ according to \cref{lemma:gidney}. Upon completion of this Toffoli gate, we revert to the first case.
\end{itemize}
Combining the above two cases, the transformation of row $j$ can be implemented using a circuit of size $\bo{\ell + \log 2m}$. Therefore, the permutation $\sigma_{i,1}$ can be implemented using a circuit of  size $\bo{n + m\ell + m\log m}$.

The final component for implementing $\sigma$ is $\sigma_{i, 2} = (0, 1) \circ (2, 3) \circ \cdots \circ (2m-2, 2m-1)$. Given that $m$ is a power of 2, the permutation $\sigma_{i, 2}$ essentially flips the first qubit when the last $n - \log 2m$ qubits are all in the state $\ket{0}$. Thus, this permutation can be realized using $2(n - \log 2m)$ X gates and a multi-controlled Toffoli gate with the last $n-\log 2m$ qubits as the control qubits and the first qubit as the target qubit. According to \cref{lemma:gidney}, the circuit size for this Toffoli gate is $\bo{n - \log 2m}$.

In summary, the circuit size for implementing any permutation $\sigma_i$ composed of $m$ pairwise disjoint transpositions is $\bo{n + m\ell + m\log m}$. As long as $m \leq \frac{\log n}{4}$, the circuit size for implementing $\sigma_i$ simplifies to $\bo{n}$ by noting that $\log 2m \leq \ell \leq 2^{2m}-1$. Let $m \coloneqq 2^{\lfloor{\log \left({\frac{\log n}{4}}\right)}\rfloor}$.
In Step 1, we decompose $\sigma$ as $\sigma = \sigma_{M-1} \circ \dots \circ \sigma_0$, where each $\sigma_i$ consists of at most $m$ pairwise disjoint transpositions and $M = \floor{\frac{\mathrm{size}(\sigma)}{m}} + \bo{\log (\min\cbra{m, \mathrm{size}(\sigma)})}$. Thus, the circuit size for implementing $\sigma$ is $\bo{\frac{n\,\mathrm{size}(\sigma)}{\log n} + n\log\min\{\,\mathrm{size}(\sigma), \log n\,\}}$.
\end{proof}

Step 1 decomposes $\sigma$ into a composition of $\cbra{\sigma_i}$ in time $\bo{d}$. Step 2 then generates the corresponding quantum circuit for each $\sigma_i$, with each requiring $\bo{n\,\mathrm{size}(\sigma_i)}$ classical preprocessing time. Hence, the total classical time complexity of \cref{lemma:sigma-restate} is $\bo{n\,\mathrm{size}(\sigma)}$.

\begin{remark}\label{remark:sigma}
    In particular, it is shown in the proof of \cref{lemma:sigma-restate} that, if the permutation $\sigma$ is composed of pairwise disjoint transpositions and $\mathrm{size}(\sigma) \leq \frac{\log n}{2}$ is a power of 2, then $\sigma$ can be implemented by a quantum circuit of size $\bo{n}$.
\end{remark}

Now, we are ready to derive the main theorem in this section. 
\begin{theorem}[Restatement of \cref{thm:withoutancillary}]
Any $n$-qubit $d$-sparse quantum state can be prepared by a quantum circuit of size $O(\frac{nd}{\log n} + n)$ without using ancillary qubits.
\end{theorem}

\begin{proof}
    The algorithm is illustrated in \cref{alg:withoutancillary}. 
    The basic idea is to first prepare a $\ceil{\log d}$-qubit dense quantum state $\sum_i \alpha_i \ket{\sigma^{-1}(q_i)}$ where $\sigma$ is a specific permutation satisfying that $0 \leq \sigma^{-1}(q_i) \leq d - 1$, and then transform it to the target state $\sum_i \alpha_i \ket{q_i}$ by applying $\sigma$. The circuit $C_1$ prepares the $\ceil{\log d}$-qubit quantum state, and therefore can be implemented with circuit size $\bo{d}$. In the following, we turn to the implementation of $\sigma$.
    
    After the first two for loops of \cref{alg:withoutancillary}, we construct a permutation $\sigma$ composed of at most $d$ disjoint transpositions.  When $d = \om{\log n \log \log n}$, according to Lemma \ref{lemma:sigma-restate}, the permutation $\sigma$ can be implemented by a circuit of size $\bo{\frac{nd}{\log n} + n}$. 
    When $d = o\rbra{\log n \log \log n}$, we need to consider it a bit more carefully. Note that we can decompose $\sigma$ as $\sigma = \sigma' \circ \sigma_{M-1} \dots \circ \sigma_1 \circ \sigma_0$, where $M = \bo{\frac{d}{\log n}}$, each $\sigma_i$ is composed of pairwise disjoint transpositions such that $\mathrm{size}(\sigma_i) \leq \frac{\log n}{2}$ is a power of 2, and $\sigma'$ is a residual permutation composed of less than $\frac{\log n}{2}$ pairwise disjoint transpositions. However, we can always replace $\sigma'$ with $\sigma''$ which is made up by complementing $\sigma'$ with some \textit{irrelevant} transpositions such that $\mathrm{size}(\sigma'')$ is a power of 2. For example, $\rbra{i, j}$ can be an irrelevant transposition if $i,j \notin \cbra{q_i}_{0 \leq i \leq d - 1} \cup \cbra{0, \dots, d-1}$. We can always find enough {irrelevant} transpositions when $d = o\rbra{\log n \log \log n}$. Denote $\sigma'' \circ \sigma_{M-1} \dots \circ \sigma_1 \circ \sigma_0$ as $\widetilde{\sigma}$, it is easily to verify that $\sum_i \alpha_i \ket{\widetilde{\sigma}\rbra{\sigma^{-1}(q_i)}} = \sum_i \alpha_i \ket{{\sigma}\rbra{\sigma^{-1}(q_i)}} = \sum_i \alpha_i \ket{q_i}$. According to \cref{remark:sigma}, $\widetilde{\sigma}$ can be implemented by a quantum circuit of size $\bo{\frac{nd}{\log n} + n}$. 
    
    In conclusion, for any $d$-sparse quantum state, \cref{alg:withoutancillary} constructs a preparation circuit without ancillary qubits, of size $\bo{\frac{nd}{\log n} + n}$.
\end{proof}

\begin{algorithm}[htbp]
\caption{SQSP without Ancillary Qubits}\label{alg:withoutancillary} 
\begin{algorithmic}[1]
    \REQUIRE {$ \mathcal{P} = \cbra*{\rbra*{\alpha_i, q_i}}_{0 \leq i \leq d - 1}$.}
    \ENSURE {A quantum circuit for preparing state $\sum_i \alpha_i \ket{q_i}$.}

    \STATE Let $\mathrm{Flag[d]}$ be a Boolean vector with all elements initialized to $0$ ;
    \STATE Let $\sigma$ be an identity permutation ;
    \FOR{$(\alpha_i, q_i)$ in $\mathcal{P}$}
    \IF{$q_i < d$}
    \STATE $\mathrm{Flag}[q_i] \xleftarrow{} 1$;
    \ENDIF
    \ENDFOR
    \FOR{$(\alpha_i, q_i)$ in $\mathcal{P}$}
    \IF{$q_i \geq d$}
    \STATE find a index $k$ such that $\mathrm{Flag}[k] = 0$;
    \STATE $\mathrm{Flag}[k] \xleftarrow{} 1$;
    $\sigma \xleftarrow{} \sigma \circ (k, q_i)$;
    \ENDIF
    \ENDFOR

    \STATE Construct a quantum circuit $C_1$ for preparing the $\ceil{\log d}$-qubit state $\sum_{i=0}^{d-1} \alpha_i \ket{\sigma^{-1}(q_i)}$ on the $\ceil{\log d}$ least significant qubits ;

    \STATE Construct a quantum circuit $C_2$ for implementing $\sigma$ (or $\widetilde{\sigma}$, see the proof of \cref{thm:withoutancillary});

    \STATE Output the circuit $C_2C_1$ ; 
\end{algorithmic}

\end{algorithm}

\section{\texorpdfstring{SQSP with $m$ Ancillary Qubits}{SQSP with m Ancillary Qubits}}\label{sec:with}

\subsection{Algorithm for SQSP}
In this section, we focus on investigating the trade-off between the number of ancillary qubits and the circuit size of SQSP. To achieve this, we introduce the following new representation of integers, whose encoding efficiency lies between binary encoding and unary encoding.
Recall that the unary encoding $e_x$ of an integer $x$ sets the $x$-th bit to $1$, while all other bits are set to $0$. 

\begin{definition}\label{def:nr}
Given an $n$-bit integer $x = \sum_{i=0}^{n-1}2^i x_i$ and another integer $r > 0$ such that $r$ divides $n$, define $x(j, k)$ for any integers $0 \leq j < k < n$ as $x(j, k) \coloneqq \sum_{i=j}^{k-1} 2^{i-j}x_i$. The $(n, r)$-unary encoding of the integer $x$ is expressed as $e_{x(n-r, n)} e_{x(n-2r, n-r)} \dots e_{x(0, r)}$. In particular, when $r = n$, the $(n, n)$-unary encoding corresponds to the standard unary encoding.
\end{definition}

For example, when $n = 8$, $r = 2$, and $x = 11011000$, the $(8, 2)$-unary encoding of $x$ is given by $e_{11}e_{01}e_{10}e_{00} = 0001\ 0100\ 0010\ 1000$.

\begin{remark}
    For convenience in presentation, we assume that $ r $ divides $ n $. However, even if $ r $ does not divide $ n $, this will not affect the correctness of any subsequent conclusions.
\end{remark}

In the following lemma, we show that any sparse quantum state can be efficiently prepared using the $(n, r)$-unary encoding.

\begin{lemma}[Restatement of \cref{lemma:nr-encoding}]
    Given two positive integers $n$ and $r$ such that $r$ divides $n$, and a set of size $d$ $\cbra{(q_i, \alpha_i)}_{0 \leq i \leq d-1}$ such that $\sum_i \abs{\alpha_i}^2 = 1$, $q_i \in \cbra{0, 1}^n$ for all $0 \leq i \leq d-1$ and $q_i \neq q_j$ for any $i \neq j$, there exists a quantum circuit preparing the following $\frac{n2^r}{r}$-qubit state
    \begin{align}
        \sum_{i=0}^{d-1} \alpha_i \ket{e_{q_i(n-r, n)}} \ket{e_{q_i(n-2r, n-r)}} \dots \ket{e_{q_i(0, r)}}
        = \sum_{i=0}^{d-1} \alpha_i \rbra*{\bigotimes_{j=0}^{\frac{n}{r}-1} \ket{e_{q_i(jr, (j+1)r)}}}.
    \end{align}
    The circuit is of size $\bo{\frac{nd}{r}}$, using one ancillary qubits.
\end{lemma}
\begin{proof}
    See \cref{app:lemma-nr}.
\end{proof}

The following lemma transforms the unary encoding to the binary encoding. 

\begin{lemma}\label{lemma:unary2binary}
    There exists a quantum circuit that converts the unary encoding of $n$-bit integers to the binary encoding, that is, achieving the following transformation for any $0 \leq i \leq 2^n-1$:
    \begin{align}
        \ket{e_i}\ket{0}^n \xrightarrow{} \ket{0}^{\otimes 2^n} \ket{i}
    \end{align}
    The circuit has a size of $\bo{n2^n}$.
\end{lemma}
\begin{proof}
    See \cref{app:lemma-unary2binary}.
\end{proof}

With the above two key lemmas established, we obtain the following theorem for a wide parameter regime.

\begin{theorem}\label{thm:withanc}
   For any $d = \om{n\log n}$ and any $m \in \sbra{\om{n^2}, \bo{\frac{nd}{\log nd} + n}}$, any $n$-qubit $d$-sparse quantum state can be prepared by a quantum circuit of size $\bo{\frac{nd}{\log (m+n)} + n}$ using $m$ ancillary qubits.
\end{theorem}

\begin{proof}
    The algorithm is presented in \cref{alg:withanc}. Let $r \coloneqq \floor{\log \frac{m}{n}}$ such that $\frac{n2^r}{r} \leq m$. In \cref{alg:withanc}, we use an ancillary register $\texttt{M}$ consisting of $\frac{n2^r}{r}$ qubits, which is further divided into $\frac{n}{r}$ sub-registers $\cbra{\texttt{M}_j}_{0\leq j \leq \frac{n}{r}-1}$ as in \cref{lemma:nr-encoding}. The target state is prepared in another register $\texttt{R}$, which is also divided into $\frac{n}{r}$ sub-registers $\cbra{\texttt{R}_j}_{0\leq j \leq \frac{n}{r}-1}$, each consisting of $r$ qubits. 
    
    Let us first show the correctness of \cref{alg:withanc}. The initial state is 
    \begin{align}
        \rbra*{\bigotimes_{j=0}^{\frac{n}{r}-1}\ket{0}^{\otimes 2^r}_{\texttt{M}_j}} \otimes \rbra*{\bigotimes_{j=0}^{\frac{n}{r}-1}\ket{0}^{\otimes r}_{\texttt{R}_j}}.
    \end{align}
    After Step 2, we have 
    \begin{align}
        \sum_{i=0}^{d-1} \alpha_i \rbra*{\bigotimes_{j=0}^{\frac{n}{r}-1} \ket{e_{q_i(jr, (j+1)r)}}_{\texttt{M}_j}}  \otimes \rbra*{\bigotimes_{j=0}^{\frac{n}{r}-1}\ket{0}^{\otimes r}_{\texttt{R}_j}}.
    \end{align}
    Note that each $\ket{e_{q_i(jr, (j+1)r)}}_{\texttt{M}_j}$ is the unary encoding of the $r$-bit integer ${q_i(jr, (j+1)r)}$. Therefore, Step 3 is to convert $\ket{e_{q_i(jr, (j+1)r)}}_{\texttt{M}_j} \ket{0}^{\otimes r}_{\texttt{R}_j}$ to $\ket{0}^{\otimes 2^r}_{\texttt{M}_j} \ket{{q_i(jr, (j+1)r)}}_{\texttt{R}_j}$ for all $0 \leq j \leq \frac{n}{r}-1$ according to \cref{lemma:unary2binary}. 
    After Step 3, omitting the ancillary qubits, the state stored in the register $\texttt{R}$ now is 
    \begin{align}
        \sum_{i=0}^{d-1} \alpha_i \rbra*{\bigotimes_{j=0}^{\frac{n}{r}-1} \ket{q_i(jr, (j+1)r)}_{\texttt{R}_j}} = \sum_{i=0}^{d-1} \alpha_i \ket{q_i}_\texttt{R},
    \end{align}
    as desired.

    Finally, we analyze the size of the circuit. According to \cref{lemma:nr-encoding} and \cref{lemma:unary2binary}, the size of Step 2 and Step 3 is $\bo{\frac{nd}{r}}$ and $\frac{n}{r} * \bo{r2^r}$, respectively. Therefore, the circuit size is $\bo{\frac{nd}{\log m - \log n} + m}$, which is $\bo{\frac{nd}{\log (m+n)} + n}$ when  $m \in \sbra{\om{n^2}, \bo{\frac{nd}{\log nd} + n}}$.
\end{proof}

\begin{algorithm}[htbp]
    \caption{SQSP with $m$ Ancillary Qubits}\label{alg:withanc} 
    \begin{algorithmic}[1]
    \REQUIRE $ \mathcal{P} = \cbra*{\rbra*{\alpha_i, q_i}}_{0\leq i \leq d-1}$ such that $d = \om{n\log n}$  and an integer $m \in \sbra{\om{n^2}, \bo{\frac{nd}{\log d}}}$.
    \ENSURE A quantum circuit for preparing state $\sum_{i=0}^{d-1} \alpha_i \ket{q_i}$ using at most $m$ ancillary qubits.
    
    \STATE Let $r \coloneqq \floor{\log \frac{m}{n}}$, and prepare the initial state $ \rbra*{\bigotimes_{j=0}^{\frac{n}{r}-1}\ket{0}^{\otimes 2^r}_{\texttt{M}_j}} \otimes \rbra*{\bigotimes_{j=0}^{\frac{n}{r}-1}\ket{0}^{\otimes r}_{\texttt{R}_j}}$;
    \STATE Prepare state $\sum_{i=0}^{d-1} \alpha_i \rbra*{\bigotimes_{j=0}^{\frac{n}{r}-1} \ket{e_{q_i(jr, (j+1)r)}}_{\texttt{M}_j}}$ on the first $\frac{n2^r}{r}$ qubits according to \cref{lemma:nr-encoding};
    \STATE Convert $\ket{e_{q_i(jr, (j+1)r)}}_{\texttt{M}_j}$ to $\ket{q_i(jr, (j+1)r)}$ on the $jr$-th to $(j+1)r - 1$-th qubits of the second register for all $j$.
    \end{algorithmic}
\end{algorithm}

The classical time to generate the quantum circuit in Steps 2 and 3 is $\bo{\frac{nd}{r}}$ and $\bo{n2^r}$, respectively. Therefore, the total classical time complexity of \cref{alg:withanc} is $\bo{\frac{nd}{\log (m+n)} + n}$, matching the circuit size.

Combining \cref{thm:withoutancillary} with \cref{thm:withanc}, we achieve the trade-off between the number of the ancillary qubits and the circuit size for any $d \geq 1$ and $m \in \bo{\frac{nd}{\log nd}+n}$.
\begin{theorem}[Restatement of \cref{thm:tradeoff}]
   For any $m \in \bo{\frac{nd}{\log nd}+n}$, any $n$-qubit $d$-sparse quantum state can be prepared by a quantum circuit of size $\bo{\frac{nd}{\log (m+n)} + n}$ using $m$ ancillary qubits.
\end{theorem}

\begin{proof}
    When $d = \om{n\log d}$ and $m \in \om{n^2}$, we resort to \cref{thm:withanc}, and the size is $\bo{\frac{nd}{\log(m+n)} + n}$, satisfying the requirement of the theorem. Otherwise, we have $d = o\rbra{n\log n}$ (implying $m \in o\rbra{n^2}$) or $m \in o\rbra{n^2}$, and we turn to the algorithm without ancillary qubits, i.e., \cref{thm:withoutancillary}, and the circuit size is $\bo{\frac{nd}{\log n} + n}$, also satisfying the requirement of the theorem. The proof is completed.
\end{proof}

\subsection{Lower Bound on the Circuit Size }

In this section, we prove  lower bounds on the circuit size  for preparing sparse quantum states. The following proof is based on the counting argument. To overcome the challenge posed by the existence of an infinite number of single-qubit quantum gates, we discretize the parameters of these gates. For the sake of clarity, we first establish the lower bound on the circuit size without any assumption, and then proceed to the case with reasonable assumptions.

\begin{theorem}[Restatement of \cref{thm:lb1}]\label{res:thm:lb1}
Given $m$ ancillary qubits available, there exist $n$-qubit $d$-sparse quantum states  such that any algorithm to prepare them requires $\om{\frac{nd}{\log(m + n) + \log d} + n}$ elementary quantum gates.
\end{theorem}

\begin{proof}
Consider the universal quantum gate set $G \coloneqq \cbra{R_x(\theta), R_y(\theta), R_z(\theta), CNOT}$, where $\theta \in [0, 2\pi)$ is a parameter, and the set of quantum states
\begin{align}
    \mathcal{D}_d \coloneqq \left\{ \ket{\psi_\mathcal{S}} \coloneqq \frac{1}{\sqrt{d}} \sum_{i \in \mathcal{S}} \ket{i} \mid \mathcal{S} \subset \{0, \dots, 2^n-1\}, \abs{\mathcal{S}} = d \right\}.
\end{align}
For any $\ket{\psi}, \ket{\phi} \in \mathcal{D}_d$ such that $\ket{\psi} \neq \ket{\phi}$, we have $\Abs{\ket{\psi} - \ket{\phi}} \geq \sqrt{\frac{2}{d}}$. Suppose the number of elementary quantum gates required to prepare a quantum state $\ket{\psi} \in \mathcal{D}_d$ in the worst case is $T$. Note that  $ T \leq cdn $ for some constant $ c > 0 $ \cite{gleinig2021efficient, malvetti2021quantum}.

Now we consider a new set of quantum gates $\widetilde{G}$ and a new quantum state preparation task. The new set $\widetilde{G}$ is constructed from $G$ by restricting the precision of the rotation angles of the single-qubit gates to $\delta \coloneqq \sqrt{\frac{1}{4c^2d^3n^2}}$. For instance, given a single-qubit gate $R_x(\theta)$ from $G$ with a parameter $\theta$,  we  set $\widetilde{\theta} = \delta \floor{\frac{\theta}{\delta}}$, acquire $R_x(\widetilde{\theta})$ in $\widetilde{G}$, and approximate $R_x(\theta)$ with $R_x(\widetilde{\theta})$. The new quantum state preparation task is to construct a circuit with $\widetilde{G}$ to prepare a quantum state $\ket{\widetilde{\psi}}$ for any given $\ket{\psi} \in \mathcal{D}_d$ such that $\|\ket{\psi} - \ket{\widetilde{\psi}}\| \leq \sqrt{\frac{1}{4d}}$. Note that $\|\ket{\phi} - \ket{\widetilde{\psi}}\| > \sqrt{\frac{1}{4d}}$ for any $\ket{\phi} \in \mathcal{D}_d$ such that $\ket{\psi} \neq \ket{\phi}$.

The circuit to prepare $\ket{\widetilde{\psi}}$ with $\widetilde{G}$ can be constructed as follows: first, we construct a circuit $U = U_{T-1} \dots U_0$ with $G$ to prepare the quantum state $\ket{\psi}$ exactly, where $U_i \in G$ or $U_i = I$; then, we replace each single-qubit gate $R_l(\theta)$ $(l\in\{x,y,z\})$ in $U$ with $R_l(\widetilde{\theta})$ from $\widetilde{G}$. It can be verified that this new circuit $\widetilde{U}$ prepare an approximate quantum state $\ket{\widetilde{\psi}}$ for $\ket{\psi}$, satisfying $\|\ket{\psi} - \ket{\widetilde{\psi}}\| \leq \sqrt{\frac{1}{4d}}$. This demonstrates the feasibility of achieving the new quantum state preparation task with $\widetilde{G}$. Note that the discussion also reduces the task of preparing $\ket{\widetilde{\psi}}$ with $\widetilde{G}$ to preparing $\ket{\psi}$ with $G$. Therefore, the lower bound on the circuit size of the former immediately implies the lower bound on the circuit size of the latter.

Note that $\abs{\widetilde{G}} = \frac{6\pi}{\delta} + 1$. For each $U_i$, there are a total of $\rbra{\frac{6\pi}{\delta} + 2}$ choices of quantum gates (including the identity gate $I$) to select from. The position where the quantum gates act has at most  $\rbra{m + n}^2$ choices. Therefore, by the counting argument, we have
    \begin{align}
        \rbra*{\rbra{m + n}^2 \cdot \rbra{\frac{6\pi}{\delta} + 2}}^T \geq \abs{\mathcal{D}_d} \geq \rbra{\frac{2^n}{d}}^d
    \end{align}
Thus, we have $T = \om{\frac{dn - d\log d}{\log(m + n) + \log d}}$. 

In addition, we have $T = \om{d}$ from the dimensionality argument. Also,  we consider the preparation of the state $\ket{1}^{\otimes n}$ from the initial state $\ket{0}^{\otimes n}$ with arbitrary single-qubit gates and CNOT. Undoubtedly, we have to access all the qubits to prepare $\ket{1}^{\otimes n}$.  Therefore, we get a lower bound $\om{n}$ on the circuit size.

Combining the argument above, we conclude that
\begin{align}
    T = \om*{\frac{dn - d\log d}{\log(m + n) + \log d} + d + n} = \om*{\frac{nd}{\log(m + n) + \log d} + n}.
\end{align}
\end{proof}

\begin{theorem}[Restatement of \cref{thm:lb2}]
    Suppose $d \leq 2^{\epsilon n}$ for a sufficiently small constant $\epsilon \in \rbra{0, 1}$, and given $m$ ancillary qubits, if an algorithm $\mathcal{A}$ for preparing $n$-qubit $d$-sparse quantum states satisfies the following conditions:
    \begin{enumerate}
        \item $\mathcal{A}$ uses at most $\bo{d}$ single-qubit rotation gates with arbitrary angles,
        \item all other single-qubit gates amount to a total of $O(n)$ types,
    \end{enumerate}
 then $\mathcal{A}$ requires $\Omega(\frac{nd}{\log(m+n)} + n)$ elementary quantum gates in the worst case.
\end{theorem}

\begin{proof}
    Following the same discussion as in the proof of \cref{res:thm:lb1}, it remains to estimate the number of quantum circuits with  $T$ gates and containing only $\bo{d}$ single-qubit rotation gates with arbitrary angles. We distinguish $ \{R_x(\theta), R_y(\theta), R_x(\theta)\} $ from other single-qubit gates in the circuit, so there are $\bo{n}$ types of single-qubit gates in total. First, we estimate the number of circuit templates without specifying the parameters of $ \{R_x(\theta), R_y(\theta), R_x(\theta)\} $ in the circuit, and then we specify the angles of single-qubit rotation gates. Therefore, we have
\begin{align}
((m + n)^2 \cdot \bo{n})^T \cdot (\frac{6\pi}{\delta})^{\bo{d}} \geq |\mathcal{D}_d| \geq (\frac{2^n}{d})^d.
\end{align}
Hence, we have $\bo{\log (m+n)}\cdot T + \bo{d \log nd} \geq nd - d\log d$, implying $T= \Omega(\frac{nd}{\log(m + n)})$ when $d \leq 2^{\epsilon n}$ for a sufficiently small constant $\epsilon \in \rbra{0, 1}$. Similarly, combining $T = \Omega(d + n)$, we obtain the lower bound on the circuit size as $\Omega(\frac{nd}{\log(m+n)} + n)$.
\end{proof}

\section{Conclusion}\label{sec:conclusion}

In this paper, we focus on optimizing the circuit size of preparing sparse quantum states in two scenarios: without  and with ancillary qubits. First, we have demonstrated that, without ancillary qubits, any $n$-qubit $d$-sparse quantum state can be prepared by a circuit of size $O(\frac{nd}{\log n} + n)$. Second,  we have proved that with $m \in \bo{\frac{nd}{\log nd} + n}$ ancillary qubits available, the circuit size is $O(\frac{nd}{\log (m+n)} + n)$. Finally, we have established a matching lower bound  $\om{\frac{nd}{\log(m + n)} + n}$ on the circuit size for the case of $m$ ancillary qubits available when $d \leq 2^{\delta n}$ for a sufficiently small constant $\delta \in \rbra{0, 1}$ under reasonable assumptions. Additionally, we have provided a slightly weaker lower bound of $\om{\frac{nd}{\log(m + n) + \log d} + n}$ without any assumptions. Putting the above results together, we have obtained the optimal bound $\ta{\frac{nd}{\log nd} + n}$ on the circuit size in the case of unlimited  number of ancillary qubits  available.



\bibliography{lipics-v2021-sample-article}

\appendix

\section{\texorpdfstring{Proof of \cref{lemma:nr-encoding}}{Proof of Lemma 9 }}\label{app:lemma-nr}

\begin{lemma}[Restatement of \cref{lemma:nr-encoding}]
    Given two positive integers $n$ and $r$ such that $r$ divides $n$, and a set of size $d$ $\cbra{(q_i, \alpha_i)}_{0 \leq i \leq d-1}$ such that $\sum_i \abs{\alpha_i}^2 = 1$, $q_i \in \cbra{0, 1}^n$ for all $0 \leq i \leq d-1$ and $q_i \neq q_j$ for any $i \neq j$, there exists a quantum circuit preparing the following $\frac{n2^r}{r}$-qubit state
    \begin{align}
        \sum_{i=0}^{d-1} \alpha_i \ket{e_{q_i(n-r, n)}} \ket{e_{q_i(n-2r, n-r)}} \dots \ket{e_{q_i(0, r)}}
        = \sum_{i=0}^{d-1} \alpha_i \rbra*{\bigotimes_{j=0}^{\frac{n}{r}-1} \ket{e_{q_i(jr, (j+1)r)}}}.
    \end{align}
    The circuit is of size $\bo{\frac{nd}{r}}$, using one ancillary qubits.
\end{lemma}

\begin{proof}
    The algorithm is presented in \cref{alg:unary}. The basic idea is similar to that of \cite{de2020circuit, de2022double, mao2024towards}, which aimed to prepare a sparse quantum state directly using binary encoding. Let $\ket{0}_{\texttt{anc}}$ be the ancillary qubit, and let $\texttt{M}$ denote the $\frac{n2^r}{r}$-qubit working register, which is further divided into $\frac{n}{r}$ sub-register $\cbra{\texttt{M}_j}_{0\leq j \leq \frac{n}{r}-1}$, each consisted of $2^r$ qubits. For any $\alpha \in \mathbb{C}$ and $\beta \in \mathbb{R}_+$ such that $0 \leq \abs{\alpha} \leq \beta \leq 1$, define one-qubit gate $G(\alpha, \beta)$ as follows:
    \begin{align}
        G(\alpha, \beta) \coloneqq \frac{1}{\beta}
        \begin{pmatrix}
       \sqrt{\beta^2 - \abs{\alpha}^2}     & \alpha \\
    \alpha^\dagger & \sqrt{\beta^2 - \abs{\alpha}^2}
        \end{pmatrix}.
    \end{align}

    The algorithm proceeds in an iterative manner, and we argue that before the beginning of the iteration (corresponding to $i = -1$) and at the end of the $i$-th iteration for all $0 \leq i \leq d-1$, the state of the whole register is as follows:
    \begin{align}
        \ket{0}_{\texttt{anc}} \otimes \sum_{k=0}^{i}\alpha_k\rbra*{\bigotimes_{j=0}^{\frac{n}{r}-1}\ket{e_{q_k(jr, (j+1)r}}_{\texttt{M}_j}} + \sqrt{\sum_{k > i} \abs{\alpha_k}^2}\ket{1}_{\texttt{anc}} \otimes \rbra*{\bigotimes_{j=0}^{\frac{n}{r}-1}\ket{0}^{\otimes 2^r}_{\texttt{M}_j}}. \label{equ:unary}
    \end{align}
    We prove \cref{equ:unary} by induction. Before the beginning of the iteration, the initial state is $\ket{1}_{\texttt{anc}} \otimes \rbra*{\bigotimes_{j=0}^{\frac{n}{r}-1}\ket{0}^{\otimes 2^r}_{\texttt{M}_j}}$, satisfying \cref{equ:unary} for $i = -1$.
    
    In the $(i+1)$-the iteration, the left part of \cref{equ:unary}, i.e., the partial state with the ancillary qubit being $\ket{0}$, remains unchanged. After Step 3, the right part of \cref{equ:unary} is transformed to 
    \begin{align}
        \sqrt{\sum_{k > i} \abs{\alpha_k}^2}\ket{1}_{\texttt{anc}} \otimes \rbra*{\bigotimes_{j=0}^{\frac{n}{r}-1}\ket{e_{q_{i+1}(jr, (j+1)r}}_{\texttt{M}_j}},
    \end{align}
    which is transformed to, after Step 4,
    \begin{align}\label{equ:unary2}
        \alpha_{i+1}\ket{0}_{\texttt{anc}} \otimes \rbra*{\bigotimes_{j=0}^{\frac{n}{r}-1}\ket{e_{q_{i+1}(jr, (j+1)r}}_{\texttt{M}_j}} + \sqrt{\sum_{k > i+1} \abs{\alpha_k}^2}\ket{1}_{\texttt{anc}} \otimes \rbra*{\bigotimes_{j=0}^{\frac{n}{r}-1}\ket{e_{q_{i+1}(jr, (j+1)r}}_{\texttt{M}_j}}.
    \end{align}
    After Step 5, the second part of \cref{equ:unary2} is restored to $\sqrt{\sum_{k > i+1} \abs{\alpha_k}^2}\ket{1}_{\texttt{anc}} \otimes \rbra*{\bigotimes_{j=0}^{\frac{n}{r}-1}\ket{0}^{\otimes 2^r}_{\texttt{M}_j}}$. Therefore, we get the desired state as \cref{equ:unary} for $i \xleftarrow{} i + 1$.

    Finally, we turn to the circuit size of the algorithm. In each iteration, in addition to a total of $\frac{2n}{r}$ CNOT gates in Step 3 and Step 5, one $G(\alpha, \beta)$ gate conditioned on $\frac{n}{r}$ qubits are applied, which can be implemented with two one-qubit gates and one $(\frac{n}{r}+1)$-qubit Toffoli gate~\cite{barenco1995elementary}. Therefore, the circuit size of the circuit is $\bo{\frac{nd}{r}}$.
\end{proof}

\begin{algorithm}[htbp]
    \caption{SQSP using $(n, r)$-unary encoding}\label{alg:unary} 
    \begin{algorithmic}[1]
    \REQUIRE $ \mathcal{P} = \cbra*{\rbra*{\alpha_i, q_i}}_{0\leq i \leq d-1}$.
    \ENSURE A quantum circuit for preparing state $\sum_{i=0}^{d-1} \alpha_i \rbra*{\bigotimes_{j=0}^{\frac{n}{r}-1} \ket{e_{q_i(jr, (j+1)r)}}}$.

   \STATE  Prepare the initial state $\ket{1}_{\texttt{anc}} \otimes \rbra*{\bigotimes_{j=0}^{\frac{n}{r}-1}\ket{0}^{\otimes 2^r}_{\texttt{M}_j}}$;
   \FOR{$i = 0$ to $d-1$}
    \STATE Apply a total of $\frac{n}{r}$ CNOT gates to transform $\bigotimes_{j=0}^{\frac{n}{r}-1}\ket{0}^{\otimes 2^r}_{\texttt{M}_j}$ to $\bigotimes_{j=0}^{\frac{n}{r}-1}\ket{e_{q_i(jr, (j+1)r}}_{\texttt{M}_j}$ conditioned on the state of the ancillary qubit being $\ket{1}$; \label{step:cnot}
    \STATE Apply $G(\alpha_i, \sum_{j=i}^{d-1}\abs{\alpha_j}^2)$ on the ancillary qubit conditioned on the state of $\texttt{M}_j$ being $\ket{e_{q_i(jr, (j+1)r}}$ for all $0 \leq j \leq \frac{n}{r}-1$;
    \STATE Repeat Step~\ref{step:cnot};
    \ENDFOR
    \end{algorithmic}
\end{algorithm}

The following is a concrete example illustrating the workflow of \cref{alg:unary}.
\begin{example}
    Let $d = 3$, $n = 8$, $r = 2$, with $q_0 = 11011000$, $q_1 = 00011001$, and $q_2 = 10001111$. The initial state is $\ket{1} \ket{0}^{\otimes 16}$. In the first iteration, the state evolves as follows: 
    \begin{align}
        \ket{1}\ket{0}^{\otimes 16} 
        & \xrightarrow{} \ket{1} \ket{0001\ 0100\ 0010\ 1000} \\
        & \xrightarrow{} \sqrt{1-\abs{\alpha_0}^2} \ket{1} \ket{0001\ 0100\ 0010\ 1000} + \alpha_0 \ket{0}\ket{0001\ 0100\ 0010\ 1000}\\
        & \xrightarrow{} \sqrt{1-\abs{\alpha_0}^2} \ket{1} \ket{0}^{\otimes 16} + \alpha_0 \ket{0}\ket{0001\ 0100\ 0010\ 1000}.\label{equ:examp-nr}
    \end{align}
    In the second iteration, the state is further transformed:
    \begin{align}
    \cref{equ:examp-nr} 
    & \xrightarrow{} 
    \sqrt{1 - \abs{\alpha_0}^2} \ket{1} \ket{1000\ 0100\ 0010\ 0100} \notag \\
    & \phantom{\xrightarrow{}} 
    \ + \alpha_0 \ket{0} \ket{0001\ 0100\ 0010\ 1000} \\
    & \xrightarrow{} 
    \abs{\alpha_2} \ket{1} \ket{1000\ 0100\ 0010\ 0100} \notag \\
    & \phantom{\xrightarrow{}} 
    \ + \alpha_1 \ket{0} \ket{1000\ 0100\ 0010\ 0100} \notag \\
    & \phantom{\xrightarrow{}} 
    \ + \alpha_0 \ket{0} \ket{0001\ 0100\ 0010\ 1000} \\
    & \xrightarrow{} 
    \abs{\alpha_2} \ket{1} \ket{0}^{\otimes 16} \notag \\
    & \phantom{\xrightarrow{}} 
    \ + \alpha_1 \ket{0} \ket{1000\ 0100\ 0010\ 0100} \notag \\
    & \phantom{\xrightarrow{}} 
    \ + \alpha_0 \ket{0} \ket{0001\ 0100\ 0010\ 1000}.
\end{align}
The same argument applies to the third iteration.
\end{example}

\section{\texorpdfstring{Proof of \cref{lemma:unary2binary}}{Proof of Lemma 18}}\label{app:lemma-unary2binary}
\begin{lemma}[Restatement of \cref{lemma:unary2binary}]
    There exists a quantum circuit that converts the unary encoding of $n$-bit integers to the binary encoding, that is, achieving the following transformation for any $0 \leq i \leq 2^n-1$:
    \begin{align}
        \ket{e_i}\ket{0}^n \xrightarrow{} \ket{0}^{\otimes 2^n} \ket{i}
    \end{align}
    The circuit has a size of $\bo{n2^n}$.
\end{lemma}


\begin{algorithm}[htbp]
    \caption{A Quantum Circuit Converting the Unary Encoding to the Binary Encoding}\label{alg:unary2binary} 
    \begin{algorithmic}[1]
    \REQUIRE An integer $n > 0$.
    \ENSURE A quantum circuit performing $\ket{e_i}\ket{0}^{\otimes n} \xrightarrow{} \ket{0}^{\otimes 2^n} \ket{i}$.
   \FOR{$i = 0$ to $2^n-1$}
    \STATE Transform the second register from $\ket{0}^{\otimes n}$ to $\ket{i}$ conditioned on the state of the $i$-th qubit of the first register being $\ket{1}$ using at most $n$ CNOT gates;
    \STATE Apply a $(n+1)$-qubit Toffoli gate conditioned on the state of the second register being $\ket{i}$ and targeted on the  the $i$-th qubit of the first register;
    \ENDFOR
    \end{algorithmic}
\end{algorithm}

\begin{proof}
    The algorithm is presented in \cref{alg:unary2binary}. It is direct to verify the correctness of the algorithm. In each iteration, at most $n$ CNOT gates and one $n+1$-qubit Toffoli gate are applied. Therefore, the circuit size of the circuit is $\bo{n2^n}$.
\end{proof}

\end{document}